\documentclass[a4paper,danish]{article}
\usepackage[margin=3.5cm]{geometry}

\usepackage[utf8]{inputenc}
\usepackage[T1]{fontenc}
\usepackage{natbib, dsfont, enumitem, authblk,
  amssymb,soul,xcolor,amsmath,amsthm,graphicx,subcaption,verbatim,pgfplots,tikz,prodint}
\usetikzlibrary{calc,patterns,angles,quotes,automata, positioning,arrows,shapes}

\usepackage{titlesec}
\titlespacing*{\section}{0em}{2em}{0em}
\titlespacing*{\subsection}{0em}{2em}{0em}
\titlespacing*{\subsubsection}{0em}{2em}{0em}

\setlist[itemize]{topsep=0pt}
\setlist[enumerate]{topsep=0pt}

\definecolor{linkcolor}{rgb}{0, 0, 0.54}
\usepackage[colorlinks=true,allcolors=linkcolor,linktocpage=true]{hyperref} 

\usepackage{dsfont,enumitem, url}
\DeclareMathOperator{\E}{\mathbb{E}} 

\newcommand{\R}{\mathbb{R}}
\newcommand{\N}{\mathbb{N}}

\newcommand{\blank}{\makebox[1ex]{\textbf{$\cdot$}}}
\newcommand\independent{\protect\mathpalette{\protect\independenT}{\perp}}
\def\independenT#1#2{\mathrel{\rlap{$#1#2$}\mkern2mu{#1#2}}}
\renewcommand{\phi}{\varphi}
\renewcommand{\epsilon}{\varepsilon}
\newcommand*\diff{\mathop{}\!\mathrm{d}}

\newcommand\bigO{\ensuremath{\mathcal{O}}}
\newcommand{\midd}{\; \middle|\;}
\newcommand{\1}{\mathds{1}}
\usepackage{ifthen} 

\DeclareMathOperator*{\argmin}{\arg\!\min}

\newcommand{\data}{\ensuremath{\mathcal{D}}}
\newcommand{\sample}{\ensuremath{\mathcal{S}}}

\theoremstyle{plain} 
\newtheorem{theorem}{Theorem}
\numberwithin{theorem}{section}
\newtheorem*{theorem*}{Theorem}

\newtheorem{lemma}[theorem]{Lemma}
\newtheorem{proposition}[theorem]{Proposition}
\newtheorem{corollary}[theorem]{Corollary}

\theoremstyle{definition} 

\newtheorem*{assumption*}{Assumption}

\theoremstyle{remark}

\title{The joint survival super learner: A super learner for
  right-censored data}

\author[1]{Anders Munch}
\author[1]{Thomas A.~Gerds}

\affil[1]{Section of Biostatistics, University of Copenhagen}

\begin{document}

\maketitle

\begin{abstract}
  Risk prediction models are widely used to guide
  real-world decision-making in areas such as healthcare and
  economics, and they also play a key role in estimating nuisance
  parameters in semiparametric inference. The super learner is a
  machine learning framework that combines a library of prediction
  algorithms into a meta-learner using cross-validated loss. In the
  context of right-censored data, careful consideration must be given
  to both the choice of loss function and the estimation of expected
  loss.  Moreover, estimators such as inverse probability of censoring
  weighting require accurate modeling and an estimator of the
  censoring distribution. We propose a novel approach to super
  learning for survival analysis that jointly evaluates candidate
  learners for both the event-time distribution and the censoring
  distribution. Our method imposes no restrictions on the algorithms
  included in the library, accommodates competing risks, and does not
  rely on a single pre-specified estimator of the censoring
  distribution. We establish a finite-sample bound on the average
  price we pay for using cross-validation, and show that this price
  vanishes asymptotically, up to poly-logarithmic terms, provided that
  the size of the library does not grow faster than at a polynomial
  rate in the sample size. We demonstrate the practical utility of our
  method using prostate cancer data and compare it to existing super
  learner algorithms for survival analysis using synthesized data.
\end{abstract}

\textbf{Keywords}: competing risks, cross-validation, loss based estimation,
right-censored data, super learner

\section{Introduction}
\label{sec:introduction}

Accurately predicting risk from time-to-event data is a central
challenge in various research fields, such as epidemiology, economics,
and weather forecasting, with applications in clinical decision making
and policy interventions. For instance, in prostate cancer management,
clinicians often need to estimate a patient’s risk of disease
progression and mortality over time to make informed decisions about
treatment strategies such as active surveillance versus immediate
intervention. Reliable risk prediction models can help tailor care to
individual patients, avoid overtreatment, and allocate healthcare
resources more effectively. Super learning \citep{van2007super}, also
known as ensemble learning or stacked regression
\citep{wolpert1992stacked,breiman1996stacked}, provides a powerful
approach to this problem by combining multiple candidate prediction
models to reduce the risk of bias incurred by a single potentially
mispecified model. In survival analysis, a super learner may for example combine a
stack of Cox regression models with a stack of random survival forests
\citep[][Section 8.4]{gerds2021medical}. Such a strategy has recently
produced \textit{KDpredict} (\url{https://kdpredict.com/}) a model which
jointly predicts the risks of kidney failure and all-cause mortality
at multiple time horizons based on different sets of covariates
\citep{liu2024predicting}. To evaluate the prediction performance of
the learners, the super learner behind KDpredict uses inverse
probability of censoring weighting (IPCW), where the censoring
distribution is estimated under the restrictive assumption that it
does not depend on the covariates. This is a potential source of bias
which is difficult to overcome with the currently available methods.

In this article, we propose the {\it joint survival super learner}, a
new super learner designed to handle the specific challenges of
ensemble learning with right-censored data. The joint survival super
learner simultaneously learns prediction models for the event-time and
censoring distributions. The joint survival super learner is based on
a competing risks model for the observed data, in which censoring is
included as a state of its own, such that at any time it is known in
which state an individual is. We assume conditionally independent
censoring and exploit well-known relationships between the observed
data distribution on the one side and the partly unobserved
distributions of the event time and the censoring time on the
other. Learners for the event-time and censoring hazard functions are
then assessed using the integrated Brier score across all states of
the observed data. Our estimation framework thus naturally
incorporates competing risks, avoids restrictive assumptions on the
censoring distribution, and produces an estimator for the censoring
distribution. Our approach is also fully flexible with respect to the
choice of learners. The latter is in contrast to other proposals which
restrict the library of learners to specific model classes
\citep{polley2011-sl-cens,golmakani2020super}, as we discuss in more
detail in Section~\ref{sec:super-learning}.

To analyse the theoretical properties of the joint survival super
learner, we focus on the discrete super learner, which selects the
model in the library with the best estimated performance
\citep{van2007super}. We provide theoretical guarantees for the
performance of the joint survival super learner, and in particular
show that the discrete joint survival super learner is consistent
when the library of learners includes at least one consistent
learner. We also derive a finite-sample oracle inequality for the
discrete joint survival super learner. We demonstrate how to construct
a library of learners using common methods for survival analysis and
illustrate the use of the joint survival super learner using a
prostate cancer data set.

The article is organized as follows. We introduce our notation and
framework in Section~\ref{sec:framework}.
Section~\ref{sec:super-learning} introduces loss-based super learning
and discusses other existing super learners for right-censored
data. In Section~\ref{sec:joint-survival-super-learner} we define the joint
survival super learner, while Section~\ref{sec:theor-results-prop}
provides theoretical guarantees. Section~\ref{sec:numer-exper} reports
the results of a series of numerical experiments, and
Section~\ref{sec:real-data-appl} illustrates the method on prostate
cancer data. We conclude with a discussion in
Section~\ref{sec:discussion}. Proofs are collected in the
Appendix. Code and an implementation of the joint survival super
learner in R \citep{R} are available at
\url{https://github.com/amnudn/joint-survival-super-learner}.

\section{Notation and framework}
\label{sec:framework}

In a competing risks framework \citep{andersen2012statistical} with
\(J\) competing risks, let \( T\) be a time to event variable,
\(D\in\{1,2,\dots,J\}\) the cause of the event, and $X \in
\mathcal{X}$ a vector of baseline covariates taking values in a
bounded subset \( \mathcal{X} \subset \R^p \), \( p\in\N \). Let
$\tau< \infty$ be a fixed prediction horizon. We use \(\mathcal{Q} \)
to denote the collection of all probability measures on \( [0,\tau]
\times \{1,2,\dots,J\}\times \mathcal{X} \) such that \( (T, D, X) \sim Q \)
for some unknown \( Q \in \mathcal{Q} \). For \(j\in\{1,2,
\dots,J\}\), the cause-specific conditional cumulative hazard
functions \( \Lambda_{j} \colon [0, \tau] \times \mathcal{X}
\rightarrow \R_+ \) are defined as
\begin{equation*}
  \Lambda_{j}(t \mid x) = \int_0^t\frac{  Q(T \in \diff s, D=j \mid X=x )}{Q(T \geq s \mid X=x )}.
\end{equation*} For ease of presentation we assume from now on that
\(J=2\) and that the map \( t\mapsto \Lambda_j(t \mid x) \) is
continuous for all \( x \) and \( j \), however, all technical
arguments extend naturally to the general case
\citep{andersen2012statistical}.  The event-free survival function
conditional on covariates is given by
\begin{equation}
  \label{eq:surv-def}
  S(t \mid x)=\exp\left\{-\Lambda_{1}(t \mid x)-\Lambda_{2}(t \mid x)\right\}.
\end{equation}
Let \( \mathcal{M}_{\tau}\) denote the space of all conditional
cumulative hazard functions on \( [0,\tau] \times\mathcal{X}\). Any
distribution \( Q \in \mathcal{Q} \) can be characterized by
\begin{equation*}
  \label{eq:parametrizeQ}
  \begin{split}
    Q(\diff t,j,\diff x)=& \left\{S(t- \mid x)\Lambda_1(\diff t \mid x)H(\diff x)\right\}^{\1{\{j=1\}}}\\
                         &  \left\{S(t- \mid x)\Lambda_2(\diff t \mid x)H(\diff x)\right\}^{\1{\{j=2\}}},
  \end{split}
\end{equation*}
where \(\Lambda_{j} \in \mathcal{M}_{\tau}\) for \(j=1,2\) and \(H\) is the marginal
distribution of the covariates.

We consider the right-censored setting in which we observe \(O =
(\tilde{T},\tilde D, X)\), where $\tilde T = \min(T,C)$ for a
right-censoring time \(C\), $\Delta = \1{\{T \leq C\}}$, and \(\tilde
D=\Delta D\). Let \(\mathcal{P}\) denote a set of probability measures
on the sample space \(\sample = [0, \tau] \times \{0, 1, 2\}
\times \mathcal{X}\) such that \(O \sim P \) for some unknown \(P\in
\mathcal{P}\). We assume that the event times and the censoring times
are conditionally independent given covariates, \( T \independent C
\mid X \). This implies that any distribution \( P \in \mathcal{P} \)
is characterized by a distribution \( Q \in \mathcal{Q} \) and a
conditional cumulative hazard function for \( C \) given \( X \)
\citep[c.f.,][]{begun1983information,gill1997coarsening}. We use
\(\Gamma\in\mathcal{M}_{\tau}\) to denote the cumulative hazard
function of the conditional censoring distribution given
covariates. For ease of presentation we assume that \(t\mapsto
\Gamma(t \mid x) \) is continuous for all \( x \). We let
\((t,x)\mapsto G(t \mid x)=\exp\left\{-\Gamma(t \mid x)\right\}\)
denote the survival function of the conditional censoring
distribution. The distribution \( P \) is characterized by
\begin{equation}\label{eq:parametrizeP}
  \begin{split}
    P(\diff t, j, \diff x) =& \left\{G(t- \mid x)S(t- \mid x)\Lambda_1(\diff t \mid x)H(\diff x)\right\}^{\1{{\{j=1\}}}}\\
                            & \left\{G(t- \mid x)S(t- \mid x)\Lambda_2(\diff t \mid x)H(\diff x)\right\}^{\1{{\{j=2\}}}}\\
                            & \left\{G(t- \mid x)S(t- \mid x)\Gamma(\diff t \mid x)H(\diff x)\right\}^{\1{{\{j=0\}}}}\\
    = & \left\{G(t- \mid x)Q(\diff t,j,\diff x)\right\}^{\1{{\{j\ne 0\}}}}\\    
                            & \left\{G(t- \mid x)S(t- \mid x)\Gamma(\diff t \mid x)H(\diff x)\right\}^{\1{{\{j=0\}}}}.
  \end{split}
\end{equation}
Hence, we may write
\( \mathcal{P} = \{ P_{Q, \Gamma} : Q \in \mathcal{Q}, \Gamma \in
\mathcal{G} \} \) for some \( \mathcal{G} \subset \mathcal{M}_{\tau} \). We
also have \(H\)-almost everywhere
\begin{equation*}
P(\tilde T>t \mid X=x) = S(t \mid x)G(t \mid x) = \exp\left\{-\Lambda_{1}(t \mid x)-\Lambda_{2}(t \mid x)-\Gamma(t \mid x) \right\}.
\end{equation*} We assume that there exists \(\kappa<\infty\) such
that \(\Lambda_{j}(\tau- \mid x)<\kappa \), for \(j\in\{1,2\}\), and
\(\Gamma(\tau- \mid x)<\kappa\) for almost all \(x\in\mathcal
X\). This implies that \(G(\tau- \mid x)\) is bounded away
from zero for almost all \(x\in\mathcal X\).  Under these assumptions,
the conditional cumulative hazard functions \(\Lambda_{j}\) and
\(\Gamma\) can be identified from \(P\) by
\begin{align}
  \Lambda_{j}(t \mid x) &= \int_0^t\frac{  P(\tilde T \in \diff s, \tilde D=j \mid X=x )}{P(\tilde T \geq s \mid X=x )}, \label{eq:lambdaj}\\
  \Gamma(t \mid x) &= \int_0^t\frac{  P(\tilde T \in \diff s, \tilde D=0 \mid X=x )}{P(\tilde T \geq s \mid X=x )}\label{eq:gamma}.
\end{align}
Thus, we can consider $\Lambda_j$ and \(\Gamma\) as operators which map from
\( \mathcal{P} \) to \(\mathcal M_{\tau}\).

\section{Loss-based super learning}
\label{sec:super-learning}

Loss-based super learning requires a library of learners, a
cross-validation algorithm, and a loss function for evaluating
predictive performance on hold-out samples. Let \(
\data_n=\{O_i\}_{i=1}^n \in \sample^n \) be a data set of i.i.d.\
observations from \( P \in \mathcal{P} \), and $\mathcal{A}$ a
collection of candidate learners. Let \(\Theta\) be the parameter
space, which in our case is a class of functions representing
different models. Each learner \(a \in \mathcal{A}\) is a map \( a
\colon \sample^n \rightarrow \Theta \) which takes a data set as
input and returns an estimate $a(\data_n) \in \Theta$. Let \(L\colon
\Theta \times \sample \rightarrow \R_+\) be a loss function,
representing the performance of the model $\theta \in \Theta$ at the
observation \( O \in \sample \), where lower values mean better
performance.

The expected loss of a learner is estimated by splitting the data set
$\data_n$ into $K$ disjoint approximately equally sized subsets
\(\data_n^1, \data_n^2, \dots, \data_n^K \) and then calculating the
cross-validated loss
\begin{equation}
  \label{eq:cv-risk-est}
  \hat{R}_n(a; L) =
  \frac{1}{K}\sum_{k=1}^{K}
  \frac{1}{| \data_n^{k} |}\sum_{O_i \in \data_n^{k}}
  L
  {
    \left(
      a{ (\data_n^{-k})}
      , O_i
    \right)
  },
  \quad \text{with} \quad
  \data_n^{-k} = \data_n \setminus \data_n^{k}.
\end{equation}
The subset \(\data_n^{-k}\) is referred to as the \(k\)'th training
sample, while \(\data_n^{k}\) is referred to as the \(k\)'th test or
hold-out sample.
The discrete super learner is defined as
\begin{equation*}
\hat{a}_n = \argmin_{a\in\mathcal A}\hat{R}_n(a; L),
\end{equation*}
and depends on both the library of learners and the specific
partitioning of the data into cross-validation folds
\( \data_n^1, \dots, \data_n^K \).

When designing a super learner for right-censored data, particular
care must be taken in the choice of loss function and in the
estimation of the expected loss. A commonly used loss function for
right-censored data is the partial log-likelihood loss
\citep[e.g.,][]{li2016regularized,yao2017deep,lee2018deephit,katzman2018deepsurv,gensheimer2019scalable,lee2021boosted,kvamme2021continuous}.
This loss function is also recommended for super learning with
right-censored data by \cite{polley2011-sl-cens}, under the assumption
that data are observed in discrete time. However, the partial
log-likelihood loss does not work well as a general purpose measure of
performance in hold-out samples when data are observed in continuous
time. The reason is that the partial log-likelihood assigns an
infinite value to any learner that predicts piecewise constant
cumulative hazard functions, if the test set contains event times that
are not observed in the training set. This problem occurs with
prominent survival learners including the Kaplan-Meier estimator,
random survival forests, and semi-parametric Cox regression models,
and these learners cannot be included in the library of the super
learner proposed by \cite{polley2011-sl-cens}. When a proportional
hazards model is assumed, the baseline hazard function can be profiled
out of the likelihood \citep{cox1972regression}. The cross-validated
partial log-likelihood loss \citep{verweij1993cross} has therefore
been suggested as a loss function for super learning by
\cite{golmakani2020super}. However, this choice of loss function
restricts the library of learners to include only Cox proportional
hazards models, and hence excludes many learners such as, e.g., random
survival forests, additive hazards models, and accelerated failure
time models.

Alternative approaches for super learning with right-censored data use
an IPCW loss function
\citep{graf1999assessment,van2003unicv,molinaro2004tree,keles2004asymptotically,hothorn2006survival,gerds2006consistent,gonzalez2021stacked},
censoring unbiased transformations
\citep{fan1996local,steingrimsson2019censoring}, or pseudo-values
\citep{andersen2003generalised,mogensen2013random,sachs2019ensemble}.
All these methods rely on an estimator of the censoring distribution,
and their drawback is that this estimator has to be pre-specified.
Recent work by \cite{han2021inverse} and \cite{westling2021inference}
circumvents the need to pre-specify a censoring model by iterating
between estimation of the outcome and censoring models. However, this
iterative procedure is in general not guaranteed to converge to the
true data-generating mechanism
\citep[][Appendix~A.4]{munch2024thesis}.

\section{The joint survival super learner}
\label{sec:joint-survival-super-learner}

The main idea of the joint survival super learner is to specify
libraries of learners for the hazard functions \( \Lambda_1 \), \(
\Lambda_2 \), and \( \Gamma \), and to exploit the relations in
equation~(\ref{eq:parametrizeP}) to define a joint loss function. The
joint survival super learner thus evaluates a tuple of learners for \(
(\Lambda_1, \Lambda_2, \Gamma) \) based on how well they jointly
predict the observed data and the discrete joint survival super
learner chooses the best performing tuple. To formally introduce the
joint survival super learner, we define the process
\begin{equation*}
  \eta(t) = \1\{\tilde{T} \leq t, \tilde D=1\} + 2\,\1\{\tilde{T} \leq t, \tilde
  D=2\} - \1\{\tilde{T} \leq t, \tilde D=0\},
  \quad \text{for} \quad t \in [0, \tau],
\end{equation*}
which takes values in \( \{-1,0,1,2\}\). The four values represent
four mutually exclusive states. Specifically, value \( 0 \) represents
the state where the individual is still event-free and uncensored,
value \( 1\) the state where the event of interest has occurred and
was observed, value \( 2\) the state where a competing risk has
occurred and was observed, and value \( -1\) the state where the
observation is right-censored. The state occupation probabilities
given baseline covariates \( X \) are given by the function
\begin{equation}
  \label{eq:F-def}
  F(t, l, x) = P(\eta(t) = l \mid X=x),
\end{equation}
for all \( t \in [0,\tau] \), \( l \in \{-1,0,1,2\} \), and
\( x \in \mathcal{X} \).

Under conditional independent censoring, each tuple
$(\Lambda_{1}, \Lambda_{2}, \Gamma, H)$ characterizes a distribution
\(P\in\mathcal P\), c.f.\ equation~\eqref{eq:parametrizeP}, which in
turn determines \( (F, H) \). Hence, a learner for \( F \) can be
constructed from learners for \( \Lambda_1 \), \( \Lambda_2 \), and
$\Gamma$ as follows:
\begin{equation}\label{eq:transition}
  \begin{split}
    F(t, 0, x)
    &
      = P(\tilde{T} > t \mid X= x)
      =
      \exp{{\{-\Lambda_{1}(t \mid x)-\Lambda_{2}(t \mid x) - \Gamma(t \mid x)\}
      }},
    \\
    F(t, 1, x)
    &
      = P(\tilde{T} \leq t, \tilde{D}=1 \mid X=x)
      = \int_0^t F(s-, 0, x)  \Lambda_{1}(\diff s \mid x),
    \\
    F(t, 2, x)
    &
      = P(\tilde{T} \leq t, \tilde{D}=2 \mid X=x)
      = \int_0^t  F(s-, 0, x)  \Lambda_{2}(\diff s \mid x),
    \\
    F(t, -1, x)
    &
      = P(\tilde{T} \leq t, \tilde{D}=0 \mid X=x)
      = \int_0^t F(s-, 0, x)  \Gamma(\diff s \mid x).
  \end{split}
\end{equation}
Equation~\eqref{eq:transition} implies that a library for \( F \) can
by build from three libraries of learners: \(\mathcal{A}_1\),
\( \mathcal{A}_2 \), and \( \mathcal{B} \), where \(\mathcal{A}_1\)
and \( \mathcal{A}_2\) contain learners for the conditional
cause-specific cumulative hazard functions \(\Lambda_1\) and
\( \Lambda_2\), respectively, and \(\mathcal{B}\) contains learners
for the conditional cumulative hazard function of the censoring
distribution.  Taking the Cartesian product of these libraries, we
obtain a library $\Phi$ of learners for \( F \):
\begin{equation}
  \label{eq:jssl-lib-def}
  \Phi(\mathcal{A}_1, \mathcal{A}_2, \mathcal{B})
  = \{ \phi_{a_1,a_2, b} : a_1 \in \mathcal{A}_1, a_2 \in \mathcal{A}_2, b \in \mathcal{B}\},
\end{equation}
where in correspondence with the relations in equation
\eqref{eq:transition},
\begin{equation}
  \begin{split}\label{eq:anti-transition}
    \phi_{a_1,a_2, b}(\data_n)(t,0,x)
  &= \exp{\{-a_1(\data_n)(s \mid x)-a_2(\data_n)(s \mid x) - b(\data_n)(s \mid
    x)\} },
  \\
  \phi_{a_1,a_2, b}(\data_n)(t,1,x)
  &= \int_0^t
    \phi_{a_1,a_2, b}(\data_n)(s-,0,x)  a_1(\data_n)(\diff s \mid x),
  \\
  \phi_{a_1,a_2, b}(\data_n)(t,2,x)
  &= \int_0^t\phi_{a_1,a_2, b}(\data_n)(s-,0,x)  a_2(\data_n)(\diff s \mid x),
  \\
  \phi_{a_1,a_2, b}(\data_n)(t,-1,x)
  &= \int_0^t \phi_{a_1,a_2, b}(\data_n)(s-,0,x)  b(\data_n)(\diff s \mid x).
  \end{split}
\end{equation}
  Notably, the libraries \( \mathcal{A}_1 \), \(
\mathcal{A}_2 \), and \( \mathcal{B} \) can be constructed using
standard software for survival analysis.  For example, in the
\texttt{R} software we can specify various ways to include covariates
in a Cox regression model and fit learners of the hazard functions
using the \texttt{survival}-package \citep{survival-package}, and we
can specify hyper parameters of a random survival forest and derive
learners of the hazard functions using the
\texttt{randomForestSRC}-package \citep{randomForestSRC}.

To evaluate how well a function \( F \) predicts the
process $\eta$ we use the integrated Brier score \citep{graf1999assessment}
evaulated at time \(\tau\):
\begin{equation*}
  \bar B_\tau(F,O) = \int_0^{\tau} \sum_{l=-1}^{2}
  \left(
      F(t,l,X) - \1{\{\eta(t)=l\}}
  \right)^2\diff t.
\end{equation*}
Here, the integrand is the average Brier score at time \(t\) across
the four states \citep{brier1950verification}. Based on a split of a
data set \(\data_n\) into $K$ disjoint approximately equally sized
subsets (c.f., Section \ref{sec:super-learning}), each learner \(
\phi_{a_1, a_2, b} \) in the library \( \Phi(\mathcal{A}_1,
\mathcal{A}_2, \mathcal{B}) \) is evaluated using the cross-validated
loss,
\begin{equation*}
  \hat{R}_{n}(\phi_{a_1,a_2,b} ; \bar{B}_{\tau}) =
  \frac{1}{K}\sum_{k=1}^{K}
  \frac{1}{| \data_n^{k} |}\sum_{O_i \in \data_n^{k}}
  \bar B_\tau
  {
    \left(
      \phi_{a_1,a_2,b}{ (\data_n^{-k})}
      , O_i
    \right)
  },
\end{equation*}
and the discrete joint survival super learner is the best performing tuple of hazard functions:
\begin{equation}\label{eq:discrete-JSLL}
  (\hat \Lambda_{1n},\hat \Lambda_{2n}, \hat \Gamma_{n})
  =  \argmin_{(a_1,a_2,b)\in \mathcal{A}_1\times\mathcal{A}_2\times\mathcal{B}}
    \hat{R}_{n}(\phi_{a_1,a_2,b} ; \bar{B}_{\tau}).
\end{equation}
  
Cause-specific risk predictions can be obtained from
the joint survival super learner \eqref{eq:discrete-JSLL} by
substituting into the well-known formula
\citep[e.g.,][]{benichou1990estimates, ozenne2017riskregression},
\begin{equation}
  \label{eq:cs-risk-def} \hat Q_n(T \leq t, D = j \mid X=x) = \int_0^t
\exp\left\{-\hat\Lambda_{1n}(u \mid x)-\hat\Lambda_{2n}(u \mid
x)\right\} \hat\Lambda_{jn}(\diff u \mid x), \quad j \in \{1,2\}.
\end{equation} Furthermore, the joint survival super learner provides an
estimator of the censoring distribution:
\begin{equation*}
 \hat G_n(T \leq t \mid X=x) = \exp\left\{-\hat\Gamma_n(t \mid x)\right\}.
\end{equation*}

\section{Theoretical guarantees}
\label{sec:theor-results-prop}

Cross-validation is the backbone of super learning and an intuitively
reasonable procedure for fair model selection without overfitting. In
this section, we adapt the work of \cite{van2003unicv} and
\cite{vaart2006oracle} and provide a theoretical justification for the
joint survival super learner in the form of a finite-sample oracle
inequality. We begin by demonstrating that minimizing the integrated
Brier score, as defined in
Section~\ref{sec:joint-survival-super-learner}, is statistically
proper, in the sense that minimization recovers the parameter of the
data-generating distribution. Together with our finite-sample oracle
inequality (Proposition~\ref{prop:oracle-prop} below), this implies
that the joint survival super learner is consistent when it is based
on a library that includes at least one consistent learner. Another
consequence of our finite-sample oracle inequality is that the joint
survival super learner converges (nearly) at the optimal rate
achievable within the library of learners. This statement is made
precise in Corollary~\ref{cor:asymp-cons} and the following
discussion. Proofs are deferred to the Appendix.

A sensible loss function should attain the minimal expected value at
the parameter corresponding to the data-generating distribution. Loss
functions with this property are called proper, and strictly proper if
the minimizer is unique \citep{gneiting2007strictly}. Absence of
properness makes it unclear why minimizing the (estimated) expected
loss is interesting.  Proposition~\ref{prop:stric-prop} states that
the integrated Brier score as defined in
Section~\ref{sec:joint-survival-super-learner} is a strictly proper
scoring rule. To establish this result, recall that the function \(F\)
implicitly depends on the data-generating probability measure
\(P\in\mathcal P\) but that this was so-far suppressed in the
notation. We now make this dependence explicit by writing \(F_P\) for
the function determined by a given \(P \in\mathcal{P}\) in accordance
with equation equation~(\ref{eq:F-def}). In the following we let \(
\mathcal{F}_{\mathcal{P}} = \{F_P : P \in \mathcal{P}\} \).

\begin{proposition} 
  \label{prop:stric-prop}
  If \(P \in\mathcal{P}\) then
  \begin{equation*}
    F_P = \argmin_{F \in \mathcal{F}_{\mathcal{P}}}
    \E_P{[\bar{B}_\tau(F, O)]}
    ,
  \end{equation*}
  for all \( l \in \{-1, 0, 1, 2 \} \), almost all
  \( t \in [0,\tau] \), and \( P \)-almost all
  \( x \in \mathcal{X} \).
\end{proposition}

The discrete joint survival super learner defined in
\eqref{eq:discrete-JSLL} provides an estimate of the function \(F\)
\begin{equation*}
  \hat{\phi}_n=\phi_{\hat \Lambda_{1n},\hat \Lambda_{2n}, \hat \Gamma_{n}}
\end{equation*} which is obtained by substituting \((\hat
\Lambda_{1n},\hat \Lambda_{2n}, \hat \Gamma_{n})\) for \((a_1,a_2,b)\) into the structural
equations \eqref{eq:anti-transition}. To evaluate the performance of
\(\hat{\phi}_n\) we benchmark it against the data-generating
distribution \( F_P \), which according to
Proposition~\ref{prop:stric-prop} has the smallest expected
loss. Another useful theoretical benchmark is the so-called oracle
learner which is the best learner included in the library of learners
and formally defined by
\begin{equation*}
  \tilde{\phi}_n
  =  \argmin_{\phi \in \Phi(\mathcal{A}_1, \mathcal{A}_2, \mathcal{B}) }
  \tilde{R}_{n}(\phi ; \bar{B}_{\tau}),
  \quad \text{with} \quad 
  \tilde{R}_n(\phi; \bar{B}_{\tau})=
  \frac{1}{K}\sum_{k=1}^{K} 
  \E_P{
    \left[
      \bar{B}_{\tau}
      {
        \left(
          \phi{ (\data_n^{-k})}
          , O
        \right)
      } 
      \midd  \data_n^{-k}
    \right]}
  ,
\end{equation*}
where we use \( \E_P \) to denote the expectation
under the distribution \( P \) for a new observation \( O \) which is
independent of the data \( \data_n^{-k} \). Like the joint survival
super learner, the oracle learner depends on the library of learners
and on the actual partition of the data, but unlike the joint survival
super learner, it also depends on the unknown data-generating
distribution. It is hence not available in practice and serves only as
a theoretical benchmark.

In the following, we equip the space \( \mathcal{F}_{\mathcal{P}} \)
with the norm
\begin{equation}
  \label{eq:norm}
  \| F \|_{P} = 
  \left\{
    \sum_{l=-1}^{2}
    \int_0^{\tau} \E_P{\left[ F(t, l, X)^2 \right]} \diff t
  \right\}^{1/2}.
\end{equation}
This norm induces a natural performance measure because
$\| F-F_P \|_{P}$ is equal to the excess risk,
\( \E_P{[\bar{B}_\tau(F, O)]} - \E_P{[\bar{B}_\tau(F_P, O)]} \), as
shown in Lemma~\ref{lemma:norm} in the Appendix. For simplicity of
presentation, we assume that all folds of the data partition have
equal size, \( |\data_n^{-k}| = n/K \) for a fixed number of folds
\( K \). We allow the number of learners to grow with \( n \) and
write
\( \Phi_n=\Phi(\mathcal{A}_{1n}, \mathcal{A}_{2n},
\mathcal{B}_n)\) as short-hand notation emphasizing the dependence on
the sample size. We now state a finite-sample inequality that bounds
the performance of the joint survival super learner relative to that
of the oracle learner.

\begin{proposition}
  \label{prop:oracle-prop}
  For all \(P\in\mathcal{P}\), \( n \in \N \), \( k \in \{1, \dots, K\} \),
  and $\delta>0$,
  \begin{align*}
    \E_{P}{\left[ \Vert \hat{\phi}_n(\data_n^{-k}) - F_P \Vert_{P}^2 \right]}
    & \leq (1 + 2\delta)
      \E_{P}{\left[ \Vert \tilde{\phi}_n(\data_n^{-k}) - F_P \Vert_{P}^2 \right]}
    \\
    & \quad
      + (1+ \delta) 16   K \tau
      \left(
      13 + \frac{12}{\delta}
      \right)
      \frac{\log(1 + |\Phi_n|)}{n}.
  \end{align*}
\end{proposition}

The expectation in Proposition~\ref{prop:oracle-prop} is
taken with respect to the product measure \( P^{n} \) for the data set
\( \data_n \). This means that we are quantifying the average
performance of the joint survival super learner across all training
data of size \(n\). A corresponding quantity was called the expected
true error rate in \cite{efron_and_tibshirani97}. As with many
finite-sample oracle inequalities, this result is of little direct
practical utility because the right-hand side depends on
data-dependent, unknown quantities. However, it does quantify how the
number of folds, the time horizon, and the number of learners in the
library can be expected to influence the performance. The result has
the following asymptotic consequences. 

\begin{corollary}
  \label{cor:asymp-cons}
  Assume that \( |\Phi_n| = \bigO(n^q)\), for some
  \( q \in \N \) and that there exists a sequence
  \( \phi_n \in \Phi_n \), \( n \in \N \), such that
  \(  \E_{P}{\left[ \Vert
      \phi_n(\data_n^{-k}) - F_P \Vert_{P}^2 \right]} = C_P +
  \bigO(n^{-\alpha}) \), for some \( \alpha\leq 1 \) and
  \( C_P \geq 0 \).
  \begin{enumerate}[label=(\alph*)]
  \item\label{item:1} If $\alpha=1$, then
    \(
    \E_{P}{\left[ \Vert
        \hat{\phi}_n(\data_n^{-k}) - F_P \Vert_{P}^2 \right]} = C_P +
    \bigO(\log(n)^{1+\epsilon}n^{-1}) \), $\forall\epsilon>0$.
  \item\label{item:2} If $\alpha<1$, then
    \(
    \E_{P}{\left[ \Vert
        \hat{\phi}_n(\data_n^{-k}) - F_P \Vert_{P}^2 \right]} = C_P +
    \bigO(n^{-\alpha}) \).
  \end{enumerate}
\end{corollary}

Proposition~\ref{prop:oracle-prop} provided a finite-sample bound
on the average price we pay for using cross-validation, and
Corollary~\ref{cor:asymp-cons} states that this price vanishes
asymptotically, up to poly-logarithmic terms, provided that the size
of the library does not grow faster than at a polynomial rate in the
sample size. The situation \( C_P=0 \) corresponds to a setting in
which the library includes a consistent learner. Cases~\ref{item:1}
and~\ref{item:2} correspond to situations where the oracle learner
achieves a parametric or non-parametric asymptotic rate of
convergence, respectively.

To illustrate the content of Corollary~\ref{cor:asymp-cons}, consider
first a situation where we use a library with an increasing number of
Cox regression models. Each of these models will achieve a parametric
rate of convergence, to possibly different least-false cumulative
hazard functions \cite{hjort1992inference}, and hence
item~\ref{item:1} of Corollary~\ref{cor:asymp-cons} states that the
joint survival super learner based on this library will achieve a
near-parametric rate of convergence.  \( C_P \) can be set to be the
distance between the data-generating distribution and the least false
model in the library, and so the joint survival super learner will
approximate the least false model in the library at a near-parametric
rate. Another situation appears if we add more flexible models to the
library, such as Cox lasso or random survival forests.  These models
typically converge at slower rates, where the fastest achievable rate
depends on the data-generating distribution.  Item~\ref{item:2} of
Corollary~\ref{cor:asymp-cons} shows that the joint survival super
learner achieves the same convergence rate as the best-performing
learner in the library, without any knowledge of the data-generating
distribution.

\section{Numerical experiments}
\label{sec:numer-exper}

The numerical experiments have two aims. The first aim is to
demonstrate that the joint survival super learner can outperform the
IPCW based discrete super learners of \citep{gonzalez2021stacked}
which pre-specify a potentially misspecified model for the censoring
mechanism. The second aim is to show that the discrete joint survival
super learner can compete and outperform the ensemble super learner
proposed by \cite{westling2021inference}.

For the numerical experiments we have synthesized the prostate cancer
data of \cite{kattan2000pretreatment} by fitting a hierarchical
structural equation model under parametric assumptions. The outcome of
interest is the time from randomization until the combined endpoint
tumour recurrence or all-cause death. Five baseline covariates are
used to predict the outcome risk: prostate-specific antigen (PSA,
ng/mL), Gleason score sum (GSS, values between 6 and 10), radiation
dose (RD), hormone therapy (HT, yes/no) and clinical stage (CS, six
values). The study was designed such that a patient's radiation dose
depended on when the patient entered the study. This in turn implies
that the time of censoring depends on the radiation dose. The data
were re-analysed in \citep{gerds2013estimating} where a sensitivity
analysis was conducted based on simulated data. Here we use the same
simulation setup, where event and censoring times are generated
according to parametric Cox-Weibull models \citep{Bender2005}
estimated from the original data, and the covariates are generated
according to either marginal Gaussian normal or binomial distributions
estimated from the original data
\citep[c.f.,][Section~4.6]{gerds2013estimating}. We refer to this
simulation setting as `dependent censoring'. We also considered a
simulation setting where data were generated in the same way, except
that censoring was generated completely independently of the
covariates. We refer to this simulation setting as `independent
censoring'.

For all super learners, we use a library consisting of three learners:
The Nelson-Aalen estimator \citep{andersen2012statistical}, a Cox
regression model with additive effects of the covariates
\citep{cox1972regression}, and a random survival forest
\citep{ishwaran2008random}. We use the same library to learn the
cumulative hazard functions of the outcome and the censoring time,
respectively. Specifically, to obtain estimates of the cumulative
censoring hazard function we fit the learners to the modified data set
\(\{(\tilde{T}_i, 1-\Delta_i, X_i)\}_{i=1}^n \) where the roles of
censoring and outcome are exchanged.

We compare the joint survival super learner to two IPCW based super
learners: The first super learner, called IPCW(Cox), uses a Cox
regression model with additive effects of the five covariates to
estimate the censoring probabilities, while the second super learner,
called IPCW(KM), uses the Kaplan-Meier estimator to estimate the
censoring probabilities. The Cox model for the censoring distribution
is thus correctly specified in both simulation settings, while the
Kaplan-Meier estimator only estimates the censoring model correctly in
the simulation setting where censoring is independent. Both IPCW super
learners are fitted using the \texttt{R}-package
\texttt{riskRegression} \citep{Gerds_Ohlendorff_Ozenne_2023}.
%
%
The IPCW super learners use the integrated Brier score up to a fixed time
horizon (36 months). The marginal risk of the event before this time horizon is
\(\approx 24.6\)\%. Under the `dependent censoring' setting the marginal
censoring probability before the time horizon is \(\approx 61.9\)\%. Under the
`independent censoring' setting the marginal censoring probability before this
time horizon is \( \approx 38.7 \)\%.

Each super learner provides a learner for the cumulative hazard
function for the outcome of interest. From the cumulative hazard
function, we obtain a risk prediction model as described in
Section~\ref{sec:joint-survival-super-learner}, with the special case
of $\Lambda_2 = 0$. We measure the performance of each super learner
by calculating the index of prediction accuracy (IPA)
\citep{kattan2018index} at a fixed time horizon (36 months) for the
risk prediction model provided by the super learner. The IPA is 1
minus the ratio between the model's Brier score and the null model's
Brier score, where the null model is the model that does not use any
covariate information. The value of IPA is approximated using a large
(\( n = 20,000 \)) independent data set of uncensored data. As a
benchmark, we calculate the performance of the risk prediction model
chosen by the oracle selector, which has the highest IPA in the
simulation setting.

The results for the first aim are shown in
Figure~\ref{fig:ipcw-fail}. We see that in the scenario where
censoring depends on the covariates, using the Kaplan-Meier estimator
to estimate the censoring probabilities provides a risk prediction
model with an IPA that is lower than the risk prediction model
provided by the joint survival super learner. The performance of the
risk prediction model selected by the joint survival super learner is
similar to the risk prediction model selected by the IPCW(Cox) super
learner which a priori uses a correctly specified model for the
censoring distribution. Both these risk prediction models are close to
the performance of the oracle, except for small sample sizes.

\begin{figure}[ht]
  \includegraphics[width=13cm]{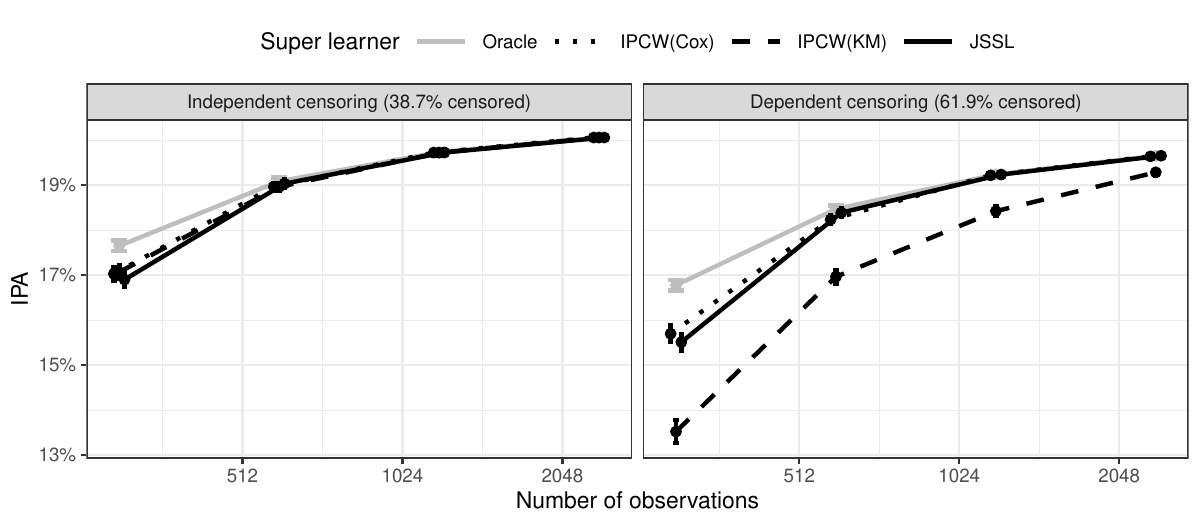}
  \caption{For the risk prediction models provided by each of the
    super learners, the IPA is plotted against sample size. The
    results are averages across 1000 simulated data sets and the error
    bars are used to quantify the Monte Carlo uncertainty. JSSL
    denotes the joint survival super learner. }
\label{fig:ipcw-fail}
\end{figure}
 
For the second aim, we consider the super learner {\it survSL}
proposed by \citep{westling2021inference} which like the joint
survival super learner does not require a pre-specified censoring
model. Both methods provide estimates of the event-time and censoring
distributions and hence we compare their performance with respect to
both the outcome and the censoring distribution. Again we use the IPA
to quantify the predictive performance.

The results for the second aim are shown in
Figures~\ref{fig:zelefski-out} and~\ref{fig:zelefski-cens}. We see
that for most sample sizes, the joint survival super learner has
similar or higher IPA compared to survSL with respect to both the
prediction of the censoring and the outcome risks. We note that the
advantage of the joint survival super learner in this particular
simulation setting might be due that it is a discrete super learner,
whereas survSL combines the learners.

\begin{figure}[ht]
  \includegraphics[width=13cm]{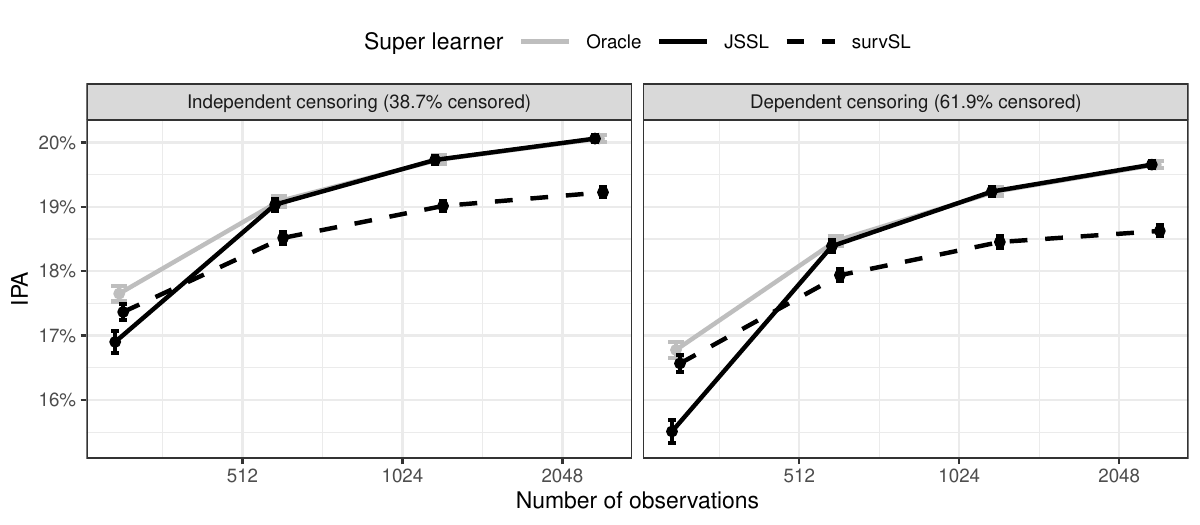}
\caption{For the risk prediction models of the outcome provided by
  each of the super learners, the IPA at the fixed time horizon is
  plotted against sample size. The results are averages across 1000
  repetitions and the error bars are used to quantify the Monte Carlo
  uncertainty. JSSL denotes the joint survival super learner.}
\label{fig:zelefski-out}
\end{figure}

\begin{figure}[ht]
  \begin{center}
      \includegraphics[width=10cm]{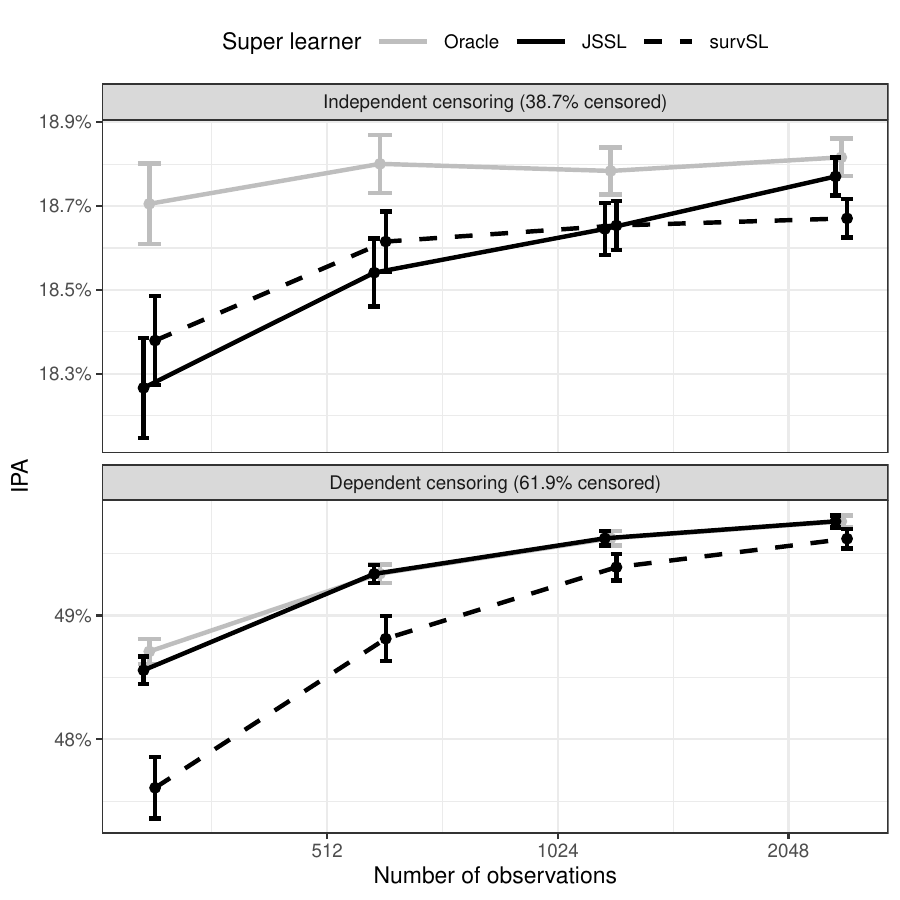}
  \end{center}
\caption{For the risk prediction models of the censoring model
    provided by each of the super learners, the IPA at the fixed time
    horizon is plotted against sample size. The results are averages
    across 1000 repetitions and the error bars are used to quantify
    the Monte Carlo uncertainty. JSSL denotes the joint survival super
    learner.}
\label{fig:zelefski-cens}
\end{figure}

\section{Prostate cancer study}
\label{sec:real-data-appl}

We use the prostate cancer data of \cite{kattan2000pretreatment} to
illustrate the use of the joint survival super learner in the presence
of competing risks. We have introduced the data in
Section~\ref{sec:numer-exper}. The data consist of 1,042 patients who
are followed from start of followup until tumour recurrence ($n=268$),
death without tumour recurrence ($n=29$), or censored ($n=745$),
whatever came first. We use the joint survival super learner to rank
libraries of learners for the cause-specific cumulative hazard
functions of tumour recurrence, death without tumour recurrence, and
censoring. The three libraries of learners each include five learners:
the Nelson-Aalen estimator, three Cox regression models (unpenalized,
lasso, elastic net) each including additive effects of the baseline
covariates, and a random survival forest. The Nelson-Aalen estimator
is estimated without covariates and serves as a benchmark model which
guarantees that the joint survival super learner is not worse than a
model which predicts the same probability to all individuals. The
other four learners use the five baseline covariates listed in
Section~\ref{sec:numer-exper} to predict the three cumulative hazard
functions of time to tumour recurrence \( \Lambda_1 \), time to death
without tumour recurrence \( \Lambda_2 \), and time to censoring
$\Gamma$. The resulting library consists of \( 5^3 = 125 \)
learners. We use five folds for training and testing the models,
repeat training and evaluation five times with different splits, and
obtain the discrete joint survival super learner as the combination
with the best average five-fold integrated Brier score across the five
repetitions. Table \ref{tab:1} shows integrated Brier scores for the
conditional state occupation probability function \( F \) as defined
in Section~\ref{sec:joint-survival-super-learner}, evaluated 3 years
after randomization for a selection of the 125 learners. We see that
among the five best combinations, the random survival forest is always
selected for \(\Gamma\) and that the differences for different
learners of \(\Lambda_1\) and \(\Lambda_2\) are small. To illustrate
the comparative performance of the discrete joint survival super
learner, we also split the data randomly into a training set with
\(n=658\) individuals and a test set with the remaining \(n=384\)
individuals.  We fit the discrete joint survival super learner (five
repetitions of five-fold cross-validation) and for comparison
pre-specified cause-specific Cox regression models to the training
set. In the training set, the discrete joint survival super learner
chooses the unpenalized Cox regression model for tumor recurrence, the
elastic net Cox regression for death without tumor recurrence, and the
random survival forest for censoring. Figure \ref{fig:zelefski-real}
compares the 3-year risk predictions from the pre-specified Cox model
and the discrete joint survival super learner in the test set.

\begin{table}[ht]
  \caption{The results of applying the 125 combinations of learners to
the prostate cancer data set. Shown are the best 5 combinations and
selected intermediate ranks.  The `Loss' is the integrated Brier
score evaluated at 3 years and `SD' is the standard deviation
across the five repetitions of five-fold cross-validation. The discrete
joint survival super learner chooses rank 1. }
\begin{center}  
\begin{tabular}{ l| c c c c c } 
  Rank&Cause 1&Cause 2&Censored&Loss&SD \\\hline
  $   1 $&elastic net&elastic net&random forest&$  7.0205 $&$ 0.030977 $ \\
  $   2 $&lasso&elastic net&random forest&$  7.0209 $&$ 0.031035 $ \\
  $   3 $&Cox&elastic net&random forest&$  7.0224 $&$ 0.030871 $ \\
  $   4 $&elastic net&lasso&random forest&$  7.0225 $&$ 0.030170 $ \\
  $   5 $&lasso&lasso&random forest&$  7.0228 $&$ 0.030237 $ \\
  $  25 $&random forest&random forest&lasso&$  7.3845 $&$ 0.026975 $ \\
  $  50 $&elastic net&random forest&lasso&$  7.3974 $&$ 0.021290 $ \\
  $  75 $&lasso&Cox&lasso&$  7.4059 $&$ 0.024660 $ \\
  $ 100 $&Nelson-Aalen&Cox&elastic net&$  7.8300 $&$ 0.016719 $ \\
  $ 125 $&Nelson-Aalen&Nelson-Aalen&Nelson-Aalen&$ 10.3298 $&$ 0.003289 $ \\
\end{tabular}\label{tab:1}
\end{center}
\end{table}

\begin{figure}[ht]
  \includegraphics[width=14cm]{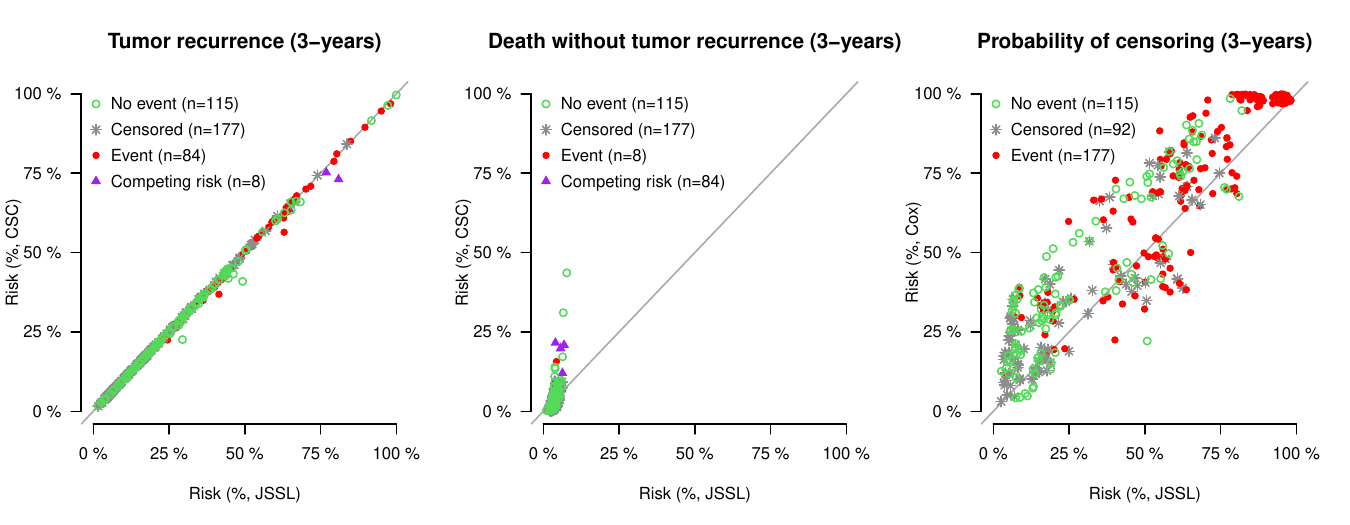}
 \caption{Comparison of risk predictions of the discrete joint
   survival super learner and pre-specified Cox regressions in an
   independent test data set.}
\label{fig:zelefski-real}
\end{figure}

\section{Discussion}
\label{sec:discussion}

A major advantage of the joint survival super learner is that the
performance of each combination of learners can be estimated without
additional nuisance parameters. A potential drawback of our approach
is that we are evaluating the loss of the learners on the level of the
observed data distribution, while the target of the analysis is either
the event-time distribution, or the censoring distribution, or both.

A relevant application of the joint survival super learner is within
the framework of targeted learning \citep{van2011targeted}, also known
as debiased machine learning \citep{chernozhukov2018double}, -- a
general methodology that combines flexible, data-adaptive estimation
of nuisance parameters with asymptotically valid inference for
low-dimensional target parameters.  For example, the methods in
\cite{van2003unified} and \cite{rytgaard2022targeted} target the
average treatment effect in a survival setting, which require
estimates of the cause-specific cumulative hazard functions and the
censoring cumulative hazard function. The joint survival super learner
can be used to estimate these nuisance parameters and is more
generally well suited for targeted and debiased machine learning with
right-censored data.

We have focused on a discrete version of the joint survival super
learner, but it is of interest to extend the method to a proper
ensemble learner, where learners are combined, e.g., through
stacking. How an ensemble should be build for tuples of learners is an
interesting topic for future research.

\appendix

\section*{Appendix}

Define
\( \bar{B}_{\tau,P}(F, o) = \bar{B}_{\tau}(F, o) -
\bar{B}_{\tau}(F_P, o) \) and
\( R_{P}(F) = \E_P{[\bar{B}_{\tau,P}(F, O)]} \), where the
integrated Brier score \( \bar{B}_{\tau} \) was defined in
Section~\ref{sec:joint-survival-super-learner}. Recall the norm
\( \Vert \blank \Vert_{P}\) defined in
equation~(\ref{eq:norm}).

\begin{lemma}
  \label{lemma:norm}
  \( R_{P}(F) = \Vert F - F_P \Vert_{P}^2 \).
\end{lemma}
\begin{proof}
  For any \( t \in [0, \tau] \) and \( l\in \{-1,0,1,2\} \) we have
  \begin{align*}
    & \E_{P}{\left[ (F(t, l, X) - \1{\{\eta(t) = l \}})^2 \right]}
    \\
    & =    \E_{P}{\left[ (F(t, l, X) - F_P(t, l, X) + F_P(t, l, X) - \1{\{\eta(t) = l
      \}})^2 \right]}
    \\
    & =    \E_{P}{\left[ (F(t, l, X) - F_P(t, l, X))^2\right]}
      + \E_{P}{\left[ (F_P(t, l, X) - \1{\{\eta(t) = l \}})^2\right]}
    \\
    & \quad
      + 2\E_{P}{\left[ (F(t, l, X) - F_P(t, l, X))(F_P(t, l, X) - \1{\{\eta(t) = l
      \}})\right]}
    \\
    & =    \E_{P}{\left[ (F(t, l, X) - F_P(t, l, X))^2\right]}
      + \E_{P}{\left[ (F_P(t, l, X) - \1{\{\eta(t) = l \}})^2\right]},
  \end{align*}
  where the last equality follows from the tower property. Hence, using Fubini,
  we have
  \begin{equation*}
    \E_P{[\bar{B}_{\tau}(F, O)]}
    = \Vert F - F_P \Vert_{P}^2 + \E_P{[\bar{B}_{\tau}(F_P, O)]}.
  \end{equation*}
\end{proof}

\begin{proof}[of Proposition~\ref{prop:stric-prop}]
  The result follows from Lemma~\ref{lemma:norm}.
\end{proof}

Recall that we use \( \Phi_n \) to denote a library of learners for the
function \( F \), and that \( \hat{\phi} \) and \( \tilde{\phi} \) denotes,
respectively, the discrete super learner and the oracle learner for the library
\( \Phi_n \), c.f., Section~\ref{sec:joint-survival-super-learner}.

\begin{proof}[of Proposition~\ref{prop:oracle-prop}]
  Minimizing the loss \( \bar{B}_{\tau} \) is equivalent to
  minimizing the loss \( \bar{B}_{\tau,P} \), so the discrete super learner and
  oracle according to \( \bar{B}_{\tau} \) and \( \bar{B}_{\tau,P} \) are
  identical. By Lemma~\ref{lemma:norm}, \( R_P(F) \geq 0 \) for any
  \( F \in \mathcal{F}_{\mathcal{P}} \), and so using Theorem 2.3 from
  \citep{vaart2006oracle} with \( p=1 \), we have that for all \( \delta >0 \),
\begin{align*}
  & \frac{1}{K} \sum_{k=1}^{K} \E_{P}{\left[ R_P(\hat{\phi}_n(\data_n^{-k})) \right]}
  \\
  &  \quad \leq
    (1+2\delta)\frac{1}{K} \sum_{k=1}^{K}\E_{P}{\left[ R_P(\tilde{\phi}_n(\data_n^{-k})) \right]}
  \\
  & \qquad + (1+\delta) \frac{16 K}{n}
    \log(1 + |\Phi_n|)\sup_{F \in \mathcal{F}_{\mathcal{P}}}
    \left\{
    M(F) + \frac{v(F)}{R_P(F)}
    \left(
    \frac{1}{\delta} + 1
    \right)
    \right\}
\end{align*}
where for each \( F \in \mathcal{F}_{\mathcal{P}} \),
\( (M(F), v(F)) \) is some Bernstein pair for the function
\(o \mapsto \bar{B}_{\tau,P}(F, o) \). As
\( \bar{B}_{\tau,P}(F, \blank) \) is uniformly bounded by \( \tau \)
for any \( F \in \mathcal{F}_{\mathcal{P}} \), it follows from section
8.1 in \citep{vaart2006oracle} that
\( (\tau, 1.5 \E_P{[\bar{B}_{\tau,P}(F, O)^2]}) \) is a Bernstein
pair for \( \bar{B}_{\tau,P}(F, \blank) \). Now, for any
\( a,b,c \in \R \) we have
\begin{align*}
  (a-c)^2 - (b-c)^2
  & = (a-b+b-c)^2 - (b-c)^2
  \\
  & = (a-b)^2 + (b-c)^2 +2(b-c)(a-b) - (b-c)^2
  \\
  & = (a-b)
    \left\{
    (a-b) +  2(b-c)
    \right\}
  \\
  & = (a-b)
    \left\{
     a + b -2c
    \right\},
\end{align*}
so using this with \( a=F(t, l, x) \), \( b=F_P(t, l, x) \), and
\( c = \1{\{\eta(t) = l\}} \), we have by Jensen's inequality
\begin{align*}
  & \E_P{[\bar{B}_{\tau,P}(F, O)^2]}
  \\
  & \leq
    2\tau\E_{P}{\left[
    \sum_{l=-1}^{2} \int_0^{\tau}
    \left\{
    \left(
    F(t, l, X) - \1{\{\eta(t) = l\}}
    \right)^2
    -
    \left(
    F_P(t, l, X) - \1{\{\eta(t) = l\}}
    \right)^2
    \right\}^2
    \diff t 
    \right]}
  \\
  & =2\tau
    \E_{P}\Bigg[
    \sum_{l=-1}^{2} \int_0^{\tau}
    \left(
    F(t, l, X) - F_P(t, l, X)
    \right)^2
  \\
  & \quad \quad \quad\quad \quad \quad \times
    \left\{
    F(t, l, X) +  F_P(t, l, X)-2 \1{\{\eta(t) = l\}}
    \right\}^2
    \diff t 
    \Bigg]
  \\
  & \leq
    8\tau \E_{P}{\left[
    \sum_{l=-1}^{2} \int_0^{\tau}
    \left(
    F(t, l, X) - F_P(t, l, X)
    \right)^2
    \diff t 
    \right]}.
  \\
  & =
    8\tau \Vert F - F_P \Vert_{P}^2.
\end{align*}
Thus when \( v(F) = 1.5 \E_P{[\bar{B}_{\tau,P}(F, O)^2]} \) we have by
Lemma~\ref{lemma:norm}
\begin{equation*}
  \frac{v(F)}{R_P(F)}
  = 1.5 \frac{\E_P{[\bar{B}_{\tau,P}(F, O)^2]}}{\E_P{[\bar{B}_{\tau,P}(F, O)]}}
  \leq 12 \tau,
\end{equation*}
and so using the Bernstein pairs \( (\tau, 1.5 \E_P{[\bar{B}_{\tau,P}(F, O)^2]}) \) we have
\begin{equation*}
  \sup_{F \in \mathcal{F}_{\mathcal{P}}}
  \left\{
    M(F) + \frac{v(F)}{R_P(F)}
    \left(
      \frac{1}{\delta} + 1
    \right)
  \right\}
  \leq \tau
  \left(
    13 + \frac{12}{\delta}
  \right).
\end{equation*}
For all $\delta>0$ we thus have
\begin{align*}
  \frac{1}{K} \sum_{k=1}^{K} \E_{P}{\left[ R_P(\hat{\phi}_n(\data_n^{-k})) \right]}
  \leq
  &(1+2\delta)\frac{1}{K} \sum_{k=1}^{K}\E_{P}{\left[ R_P(\tilde{\phi}_n(\data_n^{-k})) \right]}
  \\
  & \quad
    + (1+\delta)\log(1 + |\Phi_n|) \tau \frac{16 K}{n}
    \left(
    13 + \frac{12}{\delta}
    \right),
\end{align*}
which is equivalent to
\begin{align*}
   \E_{P}{\left[ R_P(\hat{\phi}_n(\data_n^{-k})) \right]}
  \leq
  &(1+2\delta) \E_{P}{\left[ R_P(\tilde{\phi}_n(\data_n^{-k})) \right]}
    \quad
    + (1+\delta)\log(1 + |\Phi_n|) \tau \frac{16 K}{n}
    \left(
    13 + \frac{12}{\delta}
    \right),
\end{align*}
for any \( k\in \{1, \dots, K\} \), because we have assumed that
\( |\data_n^{-k}| = n/K \) for all \( k \), and hence the expectations
on the right- and left-hand side above do not depend on \( k \). The
final result then follows from Lemma~\ref{lemma:norm}.
\end{proof}

\begin{proof}[of Corollary~\ref{cor:asymp-cons}]
  By definition of the oracle and Lemma~\ref{lemma:norm},
  \begin{equation*}
    \frac{1}{K} \sum_{k=1}^{K} \E_{P}{\left[ \Vert \tilde{\phi}_n(\data_n^{-k}) - F_P \Vert_{P}^2
      \right]} \leq
    \frac{1}{K} \sum_{k=1}^{K}\E_{P}{\left[ \Vert
        \phi_n(\data_n^{-k}) - F_P \Vert_{P}^2
      \right]}
    =
    \E_{P}{\left[ \Vert \phi_n(\data_n^{-k}) - F_P \Vert_{P}^2
      \right]},
  \end{equation*}
  for all \( n \in \N \), where the last equality follows because all
  the training sets \( \data_n^{-k} \) have the same distribution. The
  result then follows from Proposition~\ref{prop:oracle-prop} by
  letting $\delta$ grow to zero with \( n \), for instance as
  $\delta_n = \log(n)^{-\epsilon}$ for some $\epsilon>0$.
\end{proof}

\bibliography{bib.bib}

\begin{thebibliography}{54}
\providecommand{\natexlab}[1]{#1}
\providecommand{\url}[1]{\texttt{#1}}
\expandafter\ifx\csname urlstyle\endcsname\relax
  \providecommand{\doi}[1]{doi: #1}\else
  \providecommand{\doi}{doi: \begingroup \urlstyle{rm}\Url}\fi

\bibitem[Andersen et~al.(2003)Andersen, Klein, and
  Rosth{\o}j]{andersen2003generalised}
P.~K. Andersen, J.~P. Klein, and S.~Rosth{\o}j.
\newblock Generalised linear models for correlated pseudo-observations, with
  applications to multi-state models.
\newblock \emph{Biometrika}, 2003.

\bibitem[Andersen et~al.(2012)Andersen, Borgan, Gill, and
  Keiding]{andersen2012statistical}
P.~K. Andersen, O.~Borgan, R.~D. Gill, and N.~Keiding.
\newblock \emph{Statistical models based on counting processes}.
\newblock Springer Science \& Business Media, 2012.

\bibitem[Begun et~al.(1983)Begun, Hall, Huang, and
  Wellner]{begun1983information}
J.~M. Begun, W.~J. Hall, W.-M. Huang, and J.~A. Wellner.
\newblock Information and asymptotic efficiency in parametric-nonparametric
  models.
\newblock \emph{The Annals of Statistics}, 11\penalty0 (2):\penalty0 432--452,
  1983.

\bibitem[Bender et~al.(2005)Bender, Augustin, and Blettner]{Bender2005}
R.~Bender, T.~Augustin, and M.~Blettner.
\newblock {Generating survival times to simulate {Cox} proportional hazards
  models}.
\newblock \emph{Statistics in medicine}, 24:\penalty0 1713--1723, 2005.

\bibitem[Benichou and Gail(1990)]{benichou1990estimates}
J.~Benichou and M.~H. Gail.
\newblock Estimates of absolute cause-specific risk in cohort studies.
\newblock \emph{Biometrics}, pages 813--826, 1990.

\bibitem[Breiman(1996)]{breiman1996stacked}
L.~Breiman.
\newblock Stacked regressions.
\newblock \emph{Machine learning}, 24\penalty0 (1):\penalty0 49--64, 1996.

\bibitem[Brier et~al.(1950)]{brier1950verification}
G.~W. Brier et~al.
\newblock Verification of forecasts expressed in terms of probability.
\newblock \emph{Monthly weather review}, 78\penalty0 (1):\penalty0 1--3, 1950.

\bibitem[Chernozhukov et~al.(2018)Chernozhukov, Chetverikov, Demirer, Duflo,
  Hansen, Newey, and Robins]{chernozhukov2018double}
V.~Chernozhukov, D.~Chetverikov, M.~Demirer, E.~Duflo, C.~Hansen, W.~Newey, and
  J.~Robins.
\newblock Double/debiased machine learning for treatment and structural
  parameters, 2018.

\bibitem[Cox(1972)]{cox1972regression}
D.~R. Cox.
\newblock Regression models and life-tables.
\newblock \emph{Journal of the Royal Statistical Society: Series B
  (Methodological)}, 34\penalty0 (2):\penalty0 187--202, 1972.

\bibitem[Efron and Tibshirani(1997)]{efron_and_tibshirani97}
B.~Efron and R.~Tibshirani.
\newblock {Improvements on cross-validation: The .632+ bootstrap method.}
\newblock \emph{Journal of the American Statistical Association}, 92\penalty0
  (438):\penalty0 548--560, 1997.

\bibitem[Fan and Gijbels(1996)]{fan1996local}
J.~Fan and I.~Gijbels.
\newblock \emph{Local polynomial modelling and its applications}.
\newblock Routledge, 1996.

\bibitem[Gensheimer and Narasimhan(2019)]{gensheimer2019scalable}
M.~F. Gensheimer and B.~Narasimhan.
\newblock A scalable discrete-time survival model for neural networks.
\newblock \emph{PeerJ}, 7:\penalty0 e6257, 2019.

\bibitem[Gerds and Kattan(2021)]{gerds2021medical}
T.~A. Gerds and M.~W. Kattan.
\newblock \emph{Medical risk prediction models: with ties to machine learning}.
\newblock CRC Press, 2021.

\bibitem[Gerds and Schumacher(2006)]{gerds2006consistent}
T.~A. Gerds and M.~Schumacher.
\newblock Consistent estimation of the expected {B}rier score in general
  survival models with right-censored event times.
\newblock \emph{Biometrical Journal}, 48\penalty0 (6):\penalty0 1029--1040,
  2006.

\bibitem[Gerds et~al.(2013)Gerds, Kattan, Schumacher, and
  Yu]{gerds2013estimating}
T.~A. Gerds, M.~W. Kattan, M.~Schumacher, and C.~Yu.
\newblock Estimating a time-dependent concordance index for survival prediction
  models with covariate dependent censoring.
\newblock \emph{Statistics in medicine}, 32\penalty0 (13):\penalty0 2173--2184,
  2013.

\bibitem[Gerds et~al.(2023)Gerds, Ohlendorff, and
  Ozenne]{Gerds_Ohlendorff_Ozenne_2023}
T.~A. Gerds, J.~S. Ohlendorff, and B.~Ozenne.
\newblock \emph{riskRegression: Risk Regression Models and Prediction Scores
  for Survival Analysis with Competing Risks}, 2023.
\newblock URL \url{https://CRAN.R-project.org/package=riskRegression}.
\newblock R package version 2023.03.22.

\bibitem[Gill et~al.(1997)Gill, van~der Laan, and Robins]{gill1997coarsening}
R.~D. Gill, M.~J. van~der Laan, and J.~M. Robins.
\newblock Coarsening at random: Characterizations, conjectures,
  counter-examples.
\newblock In \emph{Proceedings of the First Seattle Symposium in
  Biostatistics}, pages 255--294. Springer, 1997.

\bibitem[Gneiting and Raftery(2007)]{gneiting2007strictly}
T.~Gneiting and A.~E. Raftery.
\newblock Strictly proper scoring rules, prediction, and estimation.
\newblock \emph{Journal of the American statistical Association}, 102\penalty0
  (477):\penalty0 359--378, 2007.

\bibitem[Golmakani and Polley(2020)]{golmakani2020super}
M.~K. Golmakani and E.~C. Polley.
\newblock Super learner for survival data prediction.
\newblock \emph{The International Journal of Biostatistics}, 16\penalty0
  (2):\penalty0 20190065, 2020.

\bibitem[Gonzalez~Ginestet et~al.(2021)Gonzalez~Ginestet, Kotalik, Vock,
  Wolfson, and Gabriel]{gonzalez2021stacked}
P.~Gonzalez~Ginestet, A.~Kotalik, D.~M. Vock, J.~Wolfson, and E.~E. Gabriel.
\newblock Stacked inverse probability of censoring weighted bagging: A case
  study in the infcarehiv register.
\newblock \emph{Journal of the Royal Statistical Society Series C: Applied
  Statistics}, 70\penalty0 (1):\penalty0 51--65, 2021.

\bibitem[Graf et~al.(1999)Graf, Schmoor, Sauerbrei, and
  Schumacher]{graf1999assessment}
E.~Graf, C.~Schmoor, W.~Sauerbrei, and M.~Schumacher.
\newblock Assessment and comparison of prognostic classification schemes for
  survival data.
\newblock \emph{Statistics in medicine}, 1999.

\bibitem[Han et~al.(2021)Han, Goldstein, Puli, Wies, Perotte, and
  Ranganath]{han2021inverse}
X.~Han, M.~Goldstein, A.~Puli, T.~Wies, A.~Perotte, and R.~Ranganath.
\newblock Inverse-weighted survival games.
\newblock \emph{Advances in Neural Information Processing Systems}, 34, 2021.

\bibitem[Hjort(1992)]{hjort1992inference}
N.~L. Hjort.
\newblock On inference in parametric survival data models.
\newblock \emph{International Statistical Review/Revue Internationale de
  Statistique}, pages 355--387, 1992.

\bibitem[Hothorn et~al.(2006)Hothorn, B{\"u}hlmann, Dudoit, Molinaro, and
  van~der Laan]{hothorn2006survival}
T.~Hothorn, P.~B{\"u}hlmann, S.~Dudoit, A.~Molinaro, and M.~J. van~der Laan.
\newblock Survival ensembles.
\newblock \emph{Biostatistics}, 7\penalty0 (3):\penalty0 355--373, 2006.

\bibitem[Ishwaran and Kogalur(2025)]{randomForestSRC}
H.~Ishwaran and U.~Kogalur.
\newblock \emph{Fast Unified Random Forests for Survival, Regression, and
  Classification (RF-SRC)}, 2025.
\newblock URL \url{https://cran.r-project.org/package=randomForestSRC}.
\newblock R package version 3.3.3.

\bibitem[Ishwaran et~al.(2008)Ishwaran, Kogalur, Blackstone, and
  Lauer]{ishwaran2008random}
H.~Ishwaran, U.~B. Kogalur, E.~H. Blackstone, and M.~S. Lauer.
\newblock Random survival forests.
\newblock \emph{The annals of applied statistics}, 2\penalty0 (3):\penalty0
  841--860, 2008.

\bibitem[Kattan and Gerds(2018)]{kattan2018index}
M.~W. Kattan and T.~A. Gerds.
\newblock The index of prediction accuracy: an intuitive measure useful for
  evaluating risk prediction models.
\newblock \emph{Diagnostic and prognostic research}, 2018.

\bibitem[Kattan et~al.(2000)Kattan, Zelefsky, Kupelian, Scardino, Fuks, and
  Leibel]{kattan2000pretreatment}
M.~W. Kattan, M.~J. Zelefsky, P.~A. Kupelian, P.~T. Scardino, Z.~Fuks, and
  S.~A. Leibel.
\newblock Pretreatment nomogram for predicting the outcome of three-dimensional
  conformal radiotherapy in prostate cancer.
\newblock \emph{Journal of clinical oncology}, 18\penalty0 (19):\penalty0
  3352--3359, 2000.

\bibitem[Katzman et~al.(2018)Katzman, Shaham, Cloninger, Bates, Jiang, and
  Kluger]{katzman2018deepsurv}
J.~L. Katzman, U.~Shaham, A.~Cloninger, J.~Bates, T.~Jiang, and Y.~Kluger.
\newblock Deepsurv: personalized treatment recommender system using a {C}ox
  proportional hazards deep neural network.
\newblock \emph{BMC medical research methodology}, 18\penalty0 (1):\penalty0
  1--12, 2018.

\bibitem[Keles et~al.(2004)Keles, van~der Laan, and
  Dudoit]{keles2004asymptotically}
S.~Keles, M.~van~der Laan, and S.~Dudoit.
\newblock Asymptotically optimal model selection method with right censored
  outcomes.
\newblock \emph{Bernoulli}, 10\penalty0 (6):\penalty0 1011--1037, 2004.

\bibitem[Kvamme and Borgan(2021)]{kvamme2021continuous}
H.~Kvamme and {\O}.~Borgan.
\newblock Continuous and discrete-time survival prediction with neural
  networks.
\newblock \emph{Lifetime Data Analysis}, 27\penalty0 (4):\penalty0 710--736,
  2021.

\bibitem[Lee et~al.(2018)Lee, Zame, Yoon, and van~der Schaar]{lee2018deephit}
C.~Lee, W.~Zame, J.~Yoon, and M.~van~der Schaar.
\newblock Deephit: A deep learning approach to survival analysis with competing
  risks.
\newblock In \emph{Proceedings of the AAAI conference on artificial
  intelligence}, volume~32, 2018.

\bibitem[Lee et~al.(2021)Lee, Chen, and Ishwaran]{lee2021boosted}
D.~K. Lee, N.~Chen, and H.~Ishwaran.
\newblock Boosted nonparametric hazards with time-dependent covariates.
\newblock \emph{Annals of Statistics}, 49\penalty0 (4):\penalty0 2101, 2021.

\bibitem[Li et~al.(2016)Li, Xu, and Reddy]{li2016regularized}
Y.~Li, K.~S. Xu, and C.~K. Reddy.
\newblock Regularized parametric regression for high-dimensional survival
  analysis.
\newblock In \emph{Proceedings of the 2016 SIAM International Conference on
  Data Mining}, pages 765--773. SIAM, 2016.

\bibitem[Liu et~al.(2024)Liu, Sawhney, Heide-J{\o}rgensen, Quinn, Jensen,
  Mclean, Christiansen, Gerds, and Ravani]{liu2024predicting}
P.~Liu, S.~Sawhney, U.~Heide-J{\o}rgensen, R.~R. Quinn, S.~K. Jensen,
  A.~Mclean, C.~F. Christiansen, T.~A. Gerds, and P.~Ravani.
\newblock Predicting the risks of kidney failure and death in adults with
  moderate to severe chronic kidney disease: multinational, longitudinal,
  population based, cohort study.
\newblock \emph{British Medical Journal}, 385, 2024.

\bibitem[Mogensen and Gerds(2013)]{mogensen2013random}
U.~B. Mogensen and T.~A. Gerds.
\newblock A random forest approach for competing risks based on pseudo-values.
\newblock \emph{Statistics in medicine}, 32\penalty0 (18):\penalty0 3102--3114,
  2013.

\bibitem[Molinaro et~al.(2004)Molinaro, Dudoit, and van~der
  Laan]{molinaro2004tree}
A.~M. Molinaro, S.~Dudoit, and M.~J. van~der Laan.
\newblock Tree-based multivariate regression and density estimation with
  right-censored data.
\newblock \emph{Journal of Multivariate Analysis}, 90\penalty0 (1):\penalty0
  154--177, 2004.

\bibitem[Munch(2023)]{munch2024thesis}
A.~Munch.
\newblock \emph{Targeted learning with right-censored data}.
\newblock Phd thesis, University of Copenhagen, 2023.
\newblock URL
  \url{https://publichealth.ku.dk/about-the-department/biostat/phd-theses/2023_munch.pdf}.

\bibitem[Ozenne et~al.(2017)Ozenne, S{\o}rensen, Scheike, Torp-Pedersen, and
  Gerds]{ozenne2017riskregression}
B.~Ozenne, A.~L. S{\o}rensen, T.~Scheike, C.~Torp-Pedersen, and T.~A. Gerds.
\newblock riskregression: Predicting the risk of an event using {Cox}
  regression models.
\newblock \emph{R Journal}, 9\penalty0 (2):\penalty0 440--460, 2017.

\bibitem[Polley and van~der Laan(2011)]{polley2011-sl-cens}
E.~C. Polley and M.~J. van~der Laan.
\newblock Super learning for right-censored data.
\newblock In M.~J. van~der Laan and S.~Rose, editors, \emph{Targeted Learning:
  Causal Inference for Observational and Experimental Data}, pages 249--258.
  Springer, 2011.

\bibitem[{R Core Team}(2024)]{R}
{R Core Team}.
\newblock \emph{R: A Language and Environment for Statistical Computing}.
\newblock R Foundation for Statistical Computing, Vienna, Austria, 2024.
\newblock URL \url{https://www.R-project.org/}.

\bibitem[Rytgaard and van~der Laan(2022)]{rytgaard2022targeted}
H.~C. Rytgaard and M.~J. van~der Laan.
\newblock Targeted maximum likelihood estimation for causal inference in
  survival and competing risks analysis.
\newblock \emph{Lifetime Data Analysis}, pages 1--30, 2022.

\bibitem[Sachs et~al.(2019)Sachs, Discacciati, Everhov, Ol{\'e}n, and
  Gabriel]{sachs2019ensemble}
M.~C. Sachs, A.~Discacciati, {\AA}.~H. Everhov, O.~Ol{\'e}n, and E.~E. Gabriel.
\newblock Ensemble prediction of time-to-event outcomes with competing risks: A
  case-study of surgical complications in {C}rohn's disease.
\newblock \emph{Journal of the Royal Statistical Society Series C: Applied
  Statistics}, 68\penalty0 (5):\penalty0 1431--1446, 2019.

\bibitem[Steingrimsson et~al.(2019)Steingrimsson, Diao, and
  Strawderman]{steingrimsson2019censoring}
J.~A. Steingrimsson, L.~Diao, and R.~L. Strawderman.
\newblock Censoring unbiased regression trees and ensembles.
\newblock \emph{Journal of the American Statistical Association}, 2019.

\bibitem[Therneau(2022)]{survival-package}
T.~M. Therneau.
\newblock \emph{A Package for Survival Analysis in R}, 2022.
\newblock URL \url{https://CRAN.R-project.org/package=survival}.
\newblock R package version 3.4-0.

\bibitem[van~der Laan and Dudoit(2003)]{van2003unicv}
M.~J. van~der Laan and S.~Dudoit.
\newblock Unified cross-validation methodology for selection among estimators
  and a general cross-validated adaptive epsilon-net estimator: Finite sample
  oracle inequalities and examples.
\newblock Technical report, Division of Biostatistics, University of
  California, 2003.

\bibitem[van~der Laan and Robins(2003)]{van2003unified}
M.~J. van~der Laan and J.~M. Robins.
\newblock \emph{Unified methods for censored longitudinal data and causality}.
\newblock Springer Science \& Business Media, 2003.

\bibitem[van~der Laan and Rose(2011)]{van2011targeted}
M.~J. van~der Laan and S.~Rose.
\newblock \emph{Targeted learning: causal inference for observational and
  experimental data}.
\newblock Springer Science \& Business Media, 2011.

\bibitem[van~der Laan et~al.(2007)van~der Laan, Polley, and
  Hubbard]{van2007super}
M.~J. van~der Laan, E.~C. Polley, and A.~E. Hubbard.
\newblock Super learner.
\newblock \emph{Statistical applications in genetics and molecular biology},
  6\penalty0 (1), 2007.

\bibitem[van~der Vaart et~al.(2006)van~der Vaart, Dudoit, and van~der
  Laan]{vaart2006oracle}
A.~W. van~der Vaart, S.~Dudoit, and M.~J. van~der Laan.
\newblock Oracle inequalities for multi-fold cross validation.
\newblock \emph{Statistics \& Decisions}, 24\penalty0 (3):\penalty0 351--371,
  2006.

\bibitem[Verweij and van Houwelingen(1993)]{verweij1993cross}
P.~J. Verweij and H.~C. van Houwelingen.
\newblock Cross-validation in survival analysis.
\newblock \emph{Statistics in medicine}, 12\penalty0 (24):\penalty0 2305--2314,
  1993.

\bibitem[Westling et~al.(2021)Westling, Luedtke, Gilbert, and
  Carone]{westling2021inference}
T.~Westling, A.~Luedtke, P.~Gilbert, and M.~Carone.
\newblock Inference for treatment-specific survival curves using machine
  learning.
\newblock \emph{arXiv preprint arXiv:2106.06602}, 2021.

\bibitem[Wolpert(1992)]{wolpert1992stacked}
D.~H. Wolpert.
\newblock Stacked generalization.
\newblock \emph{Neural networks}, 5\penalty0 (2):\penalty0 241--259, 1992.

\bibitem[Yao et~al.(2017)Yao, Zhu, Zhu, and Huang]{yao2017deep}
J.~Yao, X.~Zhu, F.~Zhu, and J.~Huang.
\newblock Deep correlational learning for survival prediction from
  multi-modality data.
\newblock In \emph{International conference on medical image computing and
  computer-assisted intervention}, pages 406--414. Springer, 2017.

\end{thebibliography}

\end{document}